\pdfoutput=1

\documentclass[a4paper,11pt,fleqn]{cas-sc}  

\usepackage[authoryear,longnamesfirst]{natbib}

\usepackage[colorlinks]{hyperref}
\colorlet{scolor}{black}
\colorlet{hscolor}{DarkSlateGrey}
\hypersetup{%
  pdfcreator={LaTeX3; cas-sc.cls; hyperref.sty},
  pdfproducer={pdfTeX;},
  linkcolor={hscolor},
  urlcolor={hscolor},
  citecolor={hscolor},
  filecolor={hscolor},
  menucolor={hscolor},
}
 \AtEndDocument{\hypersetup
   {pdftitle={\csuse{__short_title:}},
   pdfauthor={\csuse{__short_authors:}}}}

\usepackage{subfigure}
\usepackage{bbm}
\usepackage{mathrsfs}
\usepackage{amsfonts}
\usepackage{multirow}
\usepackage{amsmath}
\usepackage{amsthm}
\usepackage{comment}
\usepackage{amssymb}
\usepackage{graphicx}
\usepackage{tabularx}
\usepackage{setspace}
\usepackage{mathrsfs}
\usepackage{makecell,rotating,multirow,diagbox}
\usepackage{cases}
\usepackage{subeqnarray}
\usepackage{threeparttable}
\usepackage{float}

\theoremstyle{plain}
\newtheorem{theorem}{Theorem}[section]

\newtheorem{lemma}[theorem]{Lemma}

\theoremstyle{definition}

\newtheorem{assumption}{Assumption}
\theoremstyle{remark}
\newtheorem{remark}{Remark}

\usepackage{algorithm}
\usepackage{algpseudocode}

\begin{document}
\let\WriteBookmarks\relax
\def\floatpagepagefraction{1}
\def\textpagefraction{.001}
\shorttitle{Leveraging social media news}
\shortauthors{CV Radhakrishnan et~al.}

\title [mode = title]{Credit Scoring by Incorporating Dynamic Networked Information}
\tnotemark[1]
\tnotetext[1]{This research is supported by the Verg Foundation and China Scholarship Council.}

\author[1]{Yibei Li}[orcid=0000-0001-7287-1495]
\ead{yibei@kth.se}


\address[1]{Department of Mathematics, KTH Royal Institute of Technology, 10044, Stockholm, Sweden}

\author[1]{Ximei Wang}[orcid=0000-0002-4818-2910]
\cormark[1]
\ead{ximei@kth.se}

\author[1]{Boualem Djehiche}[orcid=0000-0002-6608-0715]
\ead{boualem@kth.se}



\author[1]{Xiaoming Hu}[orcid=0000-0003-0177-1993]
\ead{hu@kth.se}


\cortext[cor1]{Corresponding author}


\begin{abstract}
	In this paper,  the credit scoring problem is studied by incorporating networked information, where the advantages of such incorporation are investigated theoretically in two scenarios. Firstly, a Bayesian optimal filter is proposed to provide risk prediction for lenders assuming that published credit scores are estimated merely from structured financial data. Such prediction can then be used as a monitoring indicator for the risk management in lenders' future decisions. Secondly, a recursive Bayes estimator is further proposed to improve the precision of credit scoring by incorporating the dynamic interaction topology of clients. It is shown that under the proposed evolution framework, the designed estimator has a higher precision than any efficient estimator, and the mean square errors are strictly smaller than the Cram\'er--Rao lower bound for clients within a certain range of scores. Finally, simulation results for a special case illustrate the feasibility and effectiveness  of the proposed algorithms.
\end{abstract}

%


\maketitle
\thispagestyle{plain}

\section{Introduction}

Evaluating or estimating the credit scores according to clients' financial background  is important for lenders such as banks or other lending institutions \citep{leong2016credit, thomas2002credit}.
The credits of clients not only decide whether they can get loans or not but also influence the cost of the loans such as the lending duration and lending rates \citep{verbraken2014development,emekter2015evaluating}.
Most of the current credit scoring methods proposed in the literature are based on statistical approaches. Different kinds of models have been widely employed in order to estimate the probability of default (PD) for clients such as individuals, companies or sovereign states \citep{akkocc2012empirical,fernandes2016spatial}. 
For instance, the logistic regression model \citep{ bahnsen2014example, khemais2016credit}, ordered probit model \citep{reusens2017sovereign,tugba2019determinants}, artificial neural network (ANN) algorithm \citep{doori2014credit, byanjankar2015predicting}, support vector machines (SVM) method \citep{martens2007comprehensible, danenas2015selection}, fuzzy classification model \citep{sohn2016technology, ignatius2018fuzzy}, and set-valued identification model \citep{XimeiWangHU} are broadly used for predicting the PD for clients. \cite{lessmann2015benchmarking} compares 41 classifiers in terms of six performance measures across eight real-world credit scoring data sets and provides a holistic picture of the state-of-the-art in predictive modeling for retail scorecard development.  However,
there does not always exist the ``best'' model in credit scoring problem \citep{louzada2016classification, imtiaz2017better}. In fact, according to a large number of studies for modeling credit scoring 
there is no consistent conclusion about which method is the most accurate
\citep{migueis2013enhanced, zhang2017up}. 


The statistical approach considers information of multidimensional attributes about the client.  Usually, around 10 to 30 attributes are used as inputs in the credit scoring models. 
Such attributes typically consists of the key financial measures of clients such as their income or debt level, credit history, and payment frequency.
Moreover, different institutions may use different factors as input for the estimation of credit scores.
For some institutions in UK, credit scores of the clients are predicted by factors of payment history, age of accounts, and credit utilization. Particularly, in the UK and New Zealand for example, credit applications can be declined or delayed where there is no electoral roll listing available.
Moreover, Japan credit rating agency (JCR) evaluates credit for consumer based on factors like length of employment and salary \citep{mohammadi2016customer}. 

However, in the past few years, the power of data, algorithms and technologies have led to a dramatic change in credit scoring \citep{ntwiga2016consumer, campbell2017big}.
Networked information has attracted increasing attention since it can provide additional information for financial decision making. 
Furthermore, it is believed that  networked information about clients can be used as an effective way for improving credit ratings for clients \citep{herrero2009social, de2014network}. 
More and more studies rely on network-based data to evaluate the creditworthiness of clients.
For example, social or financial networked information helps to picture the clients in details \citep{Rusli, freedman2017information}, and network profiles of the clients such as the employment history, number of friends, or financial transfer activities may influence or determine the credit scores for each individual \citep{bolhuis2015estimating, XimeiWang}. However, there are several issues in the existing literature. Firstly, most of the studies are based on data-driven decision making algorithms \citep{Masyutin}, where various statistical algorithms are used to model the relationship between the information about the clients and the lenders' decisions from a large dataset.
For example, \cite{freedman2017information} use the data from Prosper.com, which is the largest peer-to-peer consumer lending platform to examine whether social networks facilitate online markets lending business. \cite{de2015cares}
use data from Lenddo, where they get social networked data from medias such as Facebook, Twitter, and LinkedIn to provide unique insights about clients' creditworthiness. However, it is difficult to conduct theoretical analysis on such empirical results since they lack tractable mathematical models. 
Furthermore, the existing theoretical work about networked information mainly relies on static analysis \citep{wozabal2012coupled, Wei}, which cannot deal with the scenario where individual creditworthiness and network connections evolve dynamically. Thus new paradigms are needed to bridge the gap between the existing financial practices and quantitative analysis for a changing market. 

Motivated by growing empirical studies with correlated networked data in credit scoring, the aim of this paper is to show the feasibility of such practice from a theoretical point of view. In contrast to the existing statistical algorithms which are mostly data-driven, we use a model-based framework to analyze the impact of incorporating networked information on credit scoring  quantitatively. The correlation of networked data to individual creditworthiness is established by assuming that interaction among the clients takes place dynamically according to the principle of ``homogeneous preference'' based on others' credit assessments reported by the credit scoring agency (lenders), namely, people are more likely to form financial ties with others having similar creditworthiness \citep{Wei}. Under such assumption of homophily, we will show that this correlation between network topologies and individual creditworthiness can be exploited for further improving the precision of lenders' online decision making by incorporating networked information in a dynamic framework.

Two scenarios are considered in this paper. Firstly, when the published credit scores are merely based on individual attributes, an optimal Bayesian filter is designed to predict the creditworthiness for each client based on historical observations of network connections. Such a prediction could serve as a key monitoring indicator for risk management purposes for lenders on future financial decisions.
Moreover, a recursive Bayes estimator is proposed to further improve the precision of credit scoring by incorporating the dynamic network topology. A one-step optimal estimator is given by minimizing the average risk at each time period. It is shown that under the proposed evolution framework, the designed estimator has less uncertainty than any other efficient estimators, and the mean square errors are strictly smaller than the Cram\'er--Rao lower bound for clients having scores within a certain range. Finally, a special case is studied where the true credits are assumed to be uniformly distributed. Simulation results illustrate the feasibility and effectiveness of the proposed algorithms.

The rest of this paper is organized as follows.
In Section \ref{section:Modelling}, we give an overview of the mathematical modeling for the credit scoring problem with networked information. In Section \ref{Recursive}, a Markov model is used for prediction and the corresponding recursive Bayesian filter is derived to estimate individual credit scores based on historical observations. In Section \ref{dynamic} we propose an online scoring framework to improve the scoring precision recursively.  In Section \ref{section:Numerical Simulations}, a simulation study is given to show the feasibility and efficiency of the proposed algorithm. Finally, some discussions and conclusions are given in Section \ref{section:Conclusions and Future Work}.

\section{Mathematical model of credit scoring}\label{section:Modelling}

\subsection{Network modeling}
Throughout recent decades, various kinds of networks have been regarded as informative tools to accumulate huge amounts of data and provide insights on people's behaviors. Different kinds of networks, such as social media data (blogs, Facebook and LinkedIn), trading networks and mobile phone data, have been used by the financial sector to make robust and informed decisions in the field of risk management.
There is consistent support \citep{Wei,haenlein2011social,goel2013predicting} that social or financial ties are more likely to form between clients who tend to show similarities in the revenue level or consumption behaviors. Therefore, such similarity is studied in this paper to facilitate credit scoring. 

In this paper, we consider the credit scoring problem for $N$ clients.
Let $x_i \in \mathbf{R}$ denote the true credit score of client $i~(i=1,\ldots,N)$. Without loss of generality, we assume that $x_i$ evolves according to the following linear model as a weighted sum of history creditworthiness and current assets change \footnote{A discrete model is used since in practice the credit scores are evaluated at discrete time steps. Each time step stands for a period with a specific number of days.},
\begin{equation}\label{e_model}
x_i \left(t \right) = a(t-1) x_i \left(t-1 \right) + b(t-1) u_i \left(t-1 \right) +w_i \left(t-1 \right),
\end{equation}
where at each time $t \geq 0$, $u_i(t)$ denotes the change of individual attributes for client $i$, $a(t) \in (0, 1]$ and $b(t) \in \mathbf{R}_+$ are given weights for the history assessment and updated attributes respectively, and $w_i \left(t \right) \sim \mathcal{N}\left( 0, Q_t \right) $ represents the uncertainty.

We model the financial network between clients by a time-varying graph denoted by $\mathcal{G}_t = (\mathcal{V}, \mathcal{E}_t)$, where the vertex set  $ \mathcal{V} = \lbrace 1,2,\ldots,N \rbrace $ denotes the clients in the network and $ \mathcal{E}_t \subset \mathcal{V} \times \mathcal{V}$ is the edge set at time $t$.

 At each time, client $i$ establishes financial connections based on its own credit $x_i$ and others' credit assessment $y_j$ reported by the lender or the credit scoring agency. We assume that the network is formed based on ``homogeneous preference'', i.e., clients prefer to interact people with similar credit levels. Each pair of clients meet with a probability $\nu >0$. Client $i$ forms a connection with client $j$ if and only if they have met and
\begin{equation} \label{eq:connection}
m > |x_i - y_j|,
\end{equation}
where $m$ is a random variable that models the match threshold. These two events are assumed independent.

We say that at time $t$ client $j$ is a neighbor of client $i$ if they are connected, which is denoted by $(i,j) \in \mathcal{E}_t$. The set of neighbors of client $i$ is denoted by $\mathcal{N}_{i}(t) = \lbrace j: (i,j) \in \mathcal{E}_t \rbrace$. We use $n_{i}(t)$ to denote the number of neighbors of client $i$. In the remaining part of the paper, $\lbrace g_{ij}(t) \rbrace_{i,j=1}^N$ is used to denote the elements of $\mathcal{E}_t$, where
\begin{equation*}
g_{ij}(t) = \begin{cases}
1 & \text{if}~(i,j) \in \mathcal{E}_t, \\
0 & \text{otherwise}.
\end{cases}
\end{equation*}

Hence client $i$ forms a connection with client $j$ with probability
\begin{equation}
\begin{aligned}
\text{Pr}(g_{ij}(t) = 1) &= \text{Pr}(\{i~\text{and}~j~\text{meet}\}\cap ~\{ m > |x_i - y_j\}) \\
&= \text{Pr}(m > |x_i - y_j|)\text{Pr}(i~\text{and}~j~\text{meet}),
\end{aligned}
\end{equation}
since the two events are assumed to be independent.

We determine the distribution of $m$ according to the following criteria:
\begin{enumerate}
\item[(i)] $m$ is a nonnegative random variable;
\item[(ii)] Clients are connected based on homophily preference, which means that $\text{Pr}(m > |x_i - y_j|)$ is larger with smaller credit difference $|x_i - y_j|$.
\end{enumerate}

Based on the above criteria, we choose $m$ to be Rayleigh distributed with parameter $k$ for which $P(m>x)=e^{-\frac{x^2}{2k^2}},\,\, x\ge 0$. Later in this paper we will see that such a choice  not only satisfies the aforementioned  criteria (i) and (ii), but also makes it possible to derive some useful analytic expressions in Bayesian inference. Without loss of generality, we normalize the model by taking $k=1$ and scale the credit scores to a bounded interval, i.e., $x_i(t), y_i(t) \in [0,M], \,\, M>0$, for any $t \geq 0$. Then it holds that
\begin{equation}\label{e_PrBer}
\text{Pr}(g_{ij}(t) = 1) = \nu e^{-\frac{(x_i(t) - y_j(t))^2}{2}}.
\end{equation}

Note that in most of the existing methodologies, $y_j$ is derived merely based on individual structured financial data. In the paper, we will show that the networked information can also be incorporated to achieve a higher estimation precision.

Different from \cite{Wei}, here we assume that at each time the lender only uses a partial observation of the network.
For each client $i$, the lender has an observation for its financial network by $g_{i}(t) = \lbrace g_{ij}(t) \rbrace_{j \in \mathcal{N}_{i}(t)}$.
Note that it is more reasonable to use information only from neighbors, since it is more computationally efficient.
\subsection{Problem formulation }
Up until recently, the lender assesses a client's creditworthiness  based only on its own financial history and individual attributes such as credit utilization ratio. As mentioned above, the economic engagement between clients is closely related to the homogeneity of their credits. Therefore, such a connection between the financial network and individual credits can be used to improve the scoring precision. In this paper, the following two scenarios are considered.

\begin{enumerate}
	\item \textbf{Risk prediction}
	
	The true credits of clients evolve according to the dynamics (\ref{e_model}). At each period, the lender publishes an estimated credit score for each client based only on its own information. Meanwhile, the dynamic process is observed by a risk evaluator, whose task is to provide a more precise prediction about the future creditworthiness based on historical observations of published scores and financial networks. Such prediction can then serve as a suggestion for the lender on future financial decisions.
	
	\item  \textbf{Recursive precision improvement}
	
	As a further step, the credit estimation process and the evolution of the network are considered in an integrated manner. At each period, the lender publishes an optimal estimate for each client based on the current financial network as well as individual attributes, which is then used to form a new network in the next period. The credit score estimation is coupled with the financial network in the sense that it is not only influenced by the current network, but also determines the formation of new connections in the next period.
	It is shown that the estimation precision can be improved recursively through the dynamic interaction between the lender and clients.
	
\end{enumerate}

\section{Optimal Bayesian filtering}\label{Recursive}
The risk prediction problem is studied in this section.
The credit scoring process is modeled as a Markov process, and the corresponding recursive Bayesian filter is derived to predict future creditworthiness based on historical observations of individual scores and network connections.

We assume that the lender can only get a noisy observation $y_i(t)$  of client  $i$'s credit score based on the individual information $x_i(t)$ at time $t$,
\begin{equation*}
y_i \left(t \right) = x_i \left(t \right) + v_i \left(t \right),
\end{equation*}
where $v_i \left(t \right) \sim \mathcal{N}\left( 0, R_t \right)$ denotes the observation noise.

The filtering model can be given by the following Markov process
\begin{equation}\label{e_markov}
\begin{cases}
x_i \left(t \right) = a(t-1) x_i \left(t-1 \right) + b(t-1) u_i \left(t-1 \right) +w_i \left(t-1 \right),\\
y_i \left(t \right) = x_i \left(t \right) + v_i \left(t \right), \\
g_{ij}(t) \sim Ber \left(e^{-\frac{(x_i(t)-y_j(t))^2}{2}} \right), \quad j \in \mathcal{N}_{i}(t),
\end{cases}
\end{equation}
where $w_i \left(t \right) \sim \mathcal{N}\left( 0, Q_t \right) $ and $v_i \left(t \right) \sim \mathcal{N}\left( 0, R_t \right)$ are independent Gaussian process with $R_t >0$ and $Q_t \geq 0$, respectively. We denote $ z_i(t) = [y_i(t), g_{ij}(t){\vert}_{j \in \mathcal{N}_{i}(t)}] \in \mathbf{R}^{n_i(t)+1}$ as the observation of client $i$ by the risk evaluator at the time $t$. Let $Z_{i,t} := [z_i(0),\ldots, z_i(t)]$ denote the sequence of observation history.

\begin{theorem}\label{thm_filter}
	Assume $x_i(0) \sim \mathcal{N} (\bar{x}_0, {P}_{i,0})$, and
	\begin{equation*}
	\begin{aligned}
	&\mathbb{E}[w(t)w(t+\tau)]=0, \mathbb{E}[v(t)w(t+\tau)]=0, \quad \text{for any} \quad \tau \ne 0, \\
	&\mathbb{E}[w(t){x_0}]=0, \mathbb{E}[v(t){x_0}]=0, \quad \text{for any} \quad t \geq 0.
	\end{aligned}
	\end{equation*}
	
	Then ${x_i}\left( {t} \right) | {Z_{i,t-1}} $ and ${x_i}\left( {t} \right) | {Z_{i,t}} $ are Gaussian with
	\begin{align}
	& {x_i}\left( {t } \right) | {{Z_{i,t-1}}} \sim \mathcal{N}(\hat x_i (t|t-1), P_i (t|t-1)) \label{e_filter1}\\
	& {x_i}\left( {t } \right) | {{Z_{i,t}}} \sim \mathcal{N}(\hat x_i (t|t), P_i (t|t)) \label{e_filter2}.
	\end{align}
	where
	\begin{equation}\label{e_iter}
	\begin{aligned}
	& \hat x_i (t|t-1) = a(t-1) \hat x_i (t-1|t-1) + b(t-1) u_i(t-1), \\
	& P_i (t|t-1) = a(t-1)^2 P_i (t-1|t-1) + Q_{t-1}, \\
	& \hat x_i (t|t) = \hat x_i (t|t-1) + {K_{i,t}}(y_i(t)-\hat x_i (t|t-1)) + {H_{i,t}} \sum\limits_{j \in {\mathcal{N}_{i}(t)}} (y_j(t) - \hat x_i (t|t-1)),  \\
	& P_i (t|t) = (1 - K_{i,t} - {n_{i}(t)}{H_{i,t}})P_i (t|t-1),
	\end{aligned}
	\end{equation}
	$K_{i,t}$ and ${H_{i,t}}$ are given by
	\begin{equation*}
	\begin{aligned}
	& K_{i,t} = {P_i (t|t-1)}/(R_t + P_i (t|t-1) + {n_{i}(t)}{R_t}{P_i (t|t-1)}), \\
	& H_{i,t} = {P_i (t|t-1)}{R_t}/(R_t + P_i (t|t-1) + {n_{i}(t)}{R_t}{P_i (t|t-1)}).
	\end{aligned}
	\end{equation*}
\end{theorem}

\begin{proof}
	Due to the Markov property of process (\ref{e_markov}), by Bayesian rule it holds that
	\begin{equation}\label{e_bayesrule}
	p(x_i(t+1)|Z_{i,t+1}) = \frac{p(z_i(t+1)|x_i(t+1))p(x_i(t+1)|Z_{i,t})}{p(z_i(t+1)|Z_{i,t})}.
	\end{equation}
Since $w(t)$ and $v(t)$ are Gaussian, we have
	\begin{subequations}
		\begin{align}
		& p(y_i(t)|x_i(t)) = p_{v_i}(y_i(t)-x_i(t)) = \mathcal{N} (x_i(t), R_t) \label{e_pv}\\
		& p(x_i(t+1)|x_i(t),Z_{i,t}) = p_{x_i}(x_i(t+1)-a(t)x_i(t)-b(t)u_i(t)) \notag \\
		& \hspace{101pt} = \mathcal{N} (a(t)x_i(t)+b(t)u_i(t), Q_t). \label{e_pw}
		\end{align}
	\end{subequations}
	
The Bernoulli distribution of $g_i(t)$ is given by
	\begin{equation}\label{e_GonXY}
	p(g_i(t)|x_i(t), y_i(t)) = \prod\limits_{j \in {\mathcal{N}_{i}(t)}} {{\nu e^{ - {{{{\left( {{x_i}(t) - {y_j}(t)} \right)}^2}} \over 2}}}}.
	\end{equation}
	
	Firstly, we prove (\ref{e_filter2}) by induction. To begin with, we notice that
	\begin{equation*}
	{x_i}\left( 0 \right) | Z_{i,0}  \sim \mathcal{N} (\hat x_i (0|0), P_i (0|0)),
	\end{equation*}
	with $\hat x_i (0|0) = \bar{x}_0$ and $P_i (0|0)={P}_{i,0}$.
	
	Assume that at time $t$, it holds that
	\begin{equation} \label{e_t}
	{x_i}\left( {t } \right) | {{Z_{i,t}}} \sim \mathcal{N}(\hat x_i (t|t), P_i (t|t))
	\end{equation}
	for some $\hat x_i (t|t)$ and $P_i (t|t)$.
	
	Next, we will show that the above equation also holds at time $t+1$. Notice that $p(x_i(t+1)|Z_{i,t})$ can be computed by
	\begin{equation*}
	p(x_i(t+1)|Z_{i,t}) = \int p(x_i(t+1)|x_i(t),Z_{i,t})p(x_i(t)|Z_{i,t}) d{x_i(t)}.
	\end{equation*}
	
	Then by (\ref{e_pw}) and (\ref{e_t}), we obtain
	\begin{equation} \label{e_pre}
	p(x_i(t+1)|Z_{i,t}) = \mathcal{N} (\hat x_i (t+1|t), P_i (t+1|t)),
	\end{equation}
	where
	\begin{equation*}
	\begin{aligned}
	& \hat x_i (t+1|t) = a(t) \hat x_i (t|t) + b(t) u_i(t), \\
	& P_i (t+1|t) = a(t)^2 P_i (t|t) + Q_t.
	\end{aligned}
	\end{equation*}
	
	In addition, 
	\begin{equation*}
	p(y_i(t+1), g_i(t+1)|x_i(t+1)) \!=\!
	p(y_i(t+1)|x_i(t+1))p(g_i(t+1)|y_i(t+1),x_i(t+1)\!).
	\end{equation*}
	
	Then by (\ref{e_bayesrule}), $p(x_i(t+1)|Z_{i,t+1})$ is given by
	\begin{equation*}
	\begin{split}
	p(x_i(t+1)|Z_{i,t+1}) \propto e^{-\frac{(y_i(t+1)-x_i(t+1))^2}{2R_{t+1}}}  \cdot
	e^{-\frac{(x_i(t+1)-{\hat x_i (t+1|t)})^2}{2P_i (t+1|t)}} \cdot
	\prod\limits_{j \in {N_i (t+1)}} {{e^{ - {{{{\left( {{x_i(t+1)} - {y_j(t+1)}} \right)}^2}} \over 2}}}}.
	\end{split}
	\end{equation*}
	
	It is obvious that $x_i(t+1)|Z_{i,t+1}$ is Gaussian, i.e.,
	\begin{equation*}
	x_i(t+1)|Z_{i,t+1} \sim \mathcal{N} (\hat x_i(t+1|t+1), P_i(t+1|t+1)),
	\end{equation*}
	where $x_i(t+1|t+1)$ and $P_i(t+1|t+1)$ coincide with (\ref{e_iter}).
	
	Hence we have proved (\ref{e_filter2}), and (\ref{e_filter1}) can be then derived as shown in (\ref{e_pre}).
	
\end{proof}

Based on Theorem \ref{thm_filter}, a recursive Bayesian filtering algorithm is then designed to estimate the credit scores recursively at time $t=1,...,T$. The mean-squared error (MSE) is chosen as the criterion to derive the optimal filter.
The filtering equations can then be derived based on (\ref{e_filter1}) to (\ref{e_iter}), where the MMSE estimator could be obtained as the conditional mean, i.e.,
\begin{equation*}
MMSE(\hat x_i(t)) =  \mathop {\arg \min }\limits_{\hat x_i(t)} \mathbb{E}[ {\Vert \hat x_i(t) - x_i(t) \Vert}^2 ] = \mathbb{E}[{x_i}\left( {t } \right) | {{Z_{i,t}}}] = \hat x_i(t|t).
\end{equation*}

Similar to Kalman filter, at each time step the proposed algorithm is conducted in two phases: ``predict'' and ``update'', which is given in Algorithm \ref{a1}.

%
%
%

\begin{algorithm}
	\caption{Recursive Bayesian filtering algorithm}\label{a1}
	\begin{algorithmic}[1]
		\Require {$\hat x_i (0|0) = \bar{x}_0$, $P_i (0|0)={P}_{i,0}$}
		\For{$t = 0:T$}
		\For{$i = 1:N$}
		\State \textbf{Predict:}\;
		\State $\hat x_i (t+1|t) = a(t) \hat x_i (t|t) + b(t) u_i(t)$\;
		\State $P_i (t+1|t) = a(t)^2 P_i (t|t) + Q_{t}$\;
		\State \textbf{Update:}\;
		\State $K_{i,t+1} = {P_i (t+1|t)}/(R_{t+1} + P_i (t+1|t)+ {n_{i}(t+1)}{R_{t+1}}{P_i (t+1|t)})$\;
		\State $H_{i,t+1} = {P_i (t+1|t)}{R_{t}}/(R_{t} + P_i (t+1|t)+ {n_{i}(t+1)}{R_{t+1}}{P_i (t+1|t)})$\;
		\State $\hat x_i (t+1|t+1) = \hat x_i (t+1|t) + {K_{i,t+1}}(y_i(t+1)-\hat x_i (t+1|t))$
		
		\hspace{90pt} $+ {H_{i,t+1}} \sum\limits_{j \in {\mathcal{N}_{i}(t+1)}} (y_j(t+1) - \hat x_i (t+1|t))$\;
		\State $P_i (t+1|t+1) = (1 - K_{i,t+1} - {n_{i}(t+1)}{H_{i,t+1}})P_i (t+1|t)$\;
		\EndFor
		\EndFor
	\end{algorithmic}
\end{algorithm}

It can be shown that the precision of the derived MMSE estimator is strictly higher than that of the observation by the lender which is based only on individual attributes, i.e.,
\begin{equation}\label{P_risk}
P_i (t|t) = \dfrac{R_{t} P_i (t|t-1)}{{n_{i}(t)}P_i (t|t-1)R_{t}+R_{t}+P_i (t|t-1)} < R_t.
\end{equation}

Hence, it is reasonable to adopt the proposed filter for the risk evaluator to provide a useful credit prediction for the lender. At each time, such a prediction can then be used as a reference for the lender to  make long-term financial decisions. For example, the lender may consider lowering the loan of a company if its credit score is predicted to decrease. Furthermore, since we have shown that taking the networked information into account is able to improve the scoring precision, it is natural to speculate that the lender can also incorporate the networked information in one's own assessment, which will be studied in the next section.

\section{Recursive scoring based on dynamic interaction}\label{dynamic}

\subsection{Dynamic scoring framework}
In this part, an online scoring framework is used to recursively improve the score precision based on the dynamic interaction between the lender and the clients.

As a further step, the networked data is also used by the lender when it updates the credit score for each client. During each period, the lender publishes the current score prediction, based on which the clients then form a new network with homogeneous preference. Then at the end of the period, the lender updates a new estimate for current individual credit scores on the basis of the observation of the network, which will be used to make new predictions at the beginning of the next period.

Here we use $\bar x_i(t)$ and $\hat x_i(t)$ to denote the prediction and corrected estimation for the credit scores of client $i$ at time $t$. In each period, the interaction process mentioned above is composed of the following four steps, as shown in Fig. (\ref{Fig_interaction}).
\begin{enumerate}
	\item \textbf{True credits update:}\\ The true credit scores  of the clients evolve according to the system model,
	\begin{equation}\label{e_step1}
	x_i(t) = a(t-1)x_i(t-1) + bu_i(t-1) +w_i(t-1).
	\end{equation}
	\item \textbf{{Publishing by the lender:}}\\ The lender publishes a new prediction of credit scores based on the change of individual financial attributes:
	\begin{equation}\label{e_publish}
	\bar x_i(t) = a(t-1) \hat x_i(t-1) + bu_i(t-1).
	\end{equation}
	\item \textbf{Network formation:}\\ Each client $i$ forms its part of the new network $g_i(t)$ based on its own credit $x_i(t)$ and others' scores $\bar x_j(t)$ that are reported by the lender.
	\item \textbf{Score correction:}\\ The lender updates a new credit estimate $\hat x_i(t)$ for each client based on the observation of the current network $g_i(t)$.
\end{enumerate}

\begin{figure}
	\centering
	\includegraphics[width=0.5\textwidth]{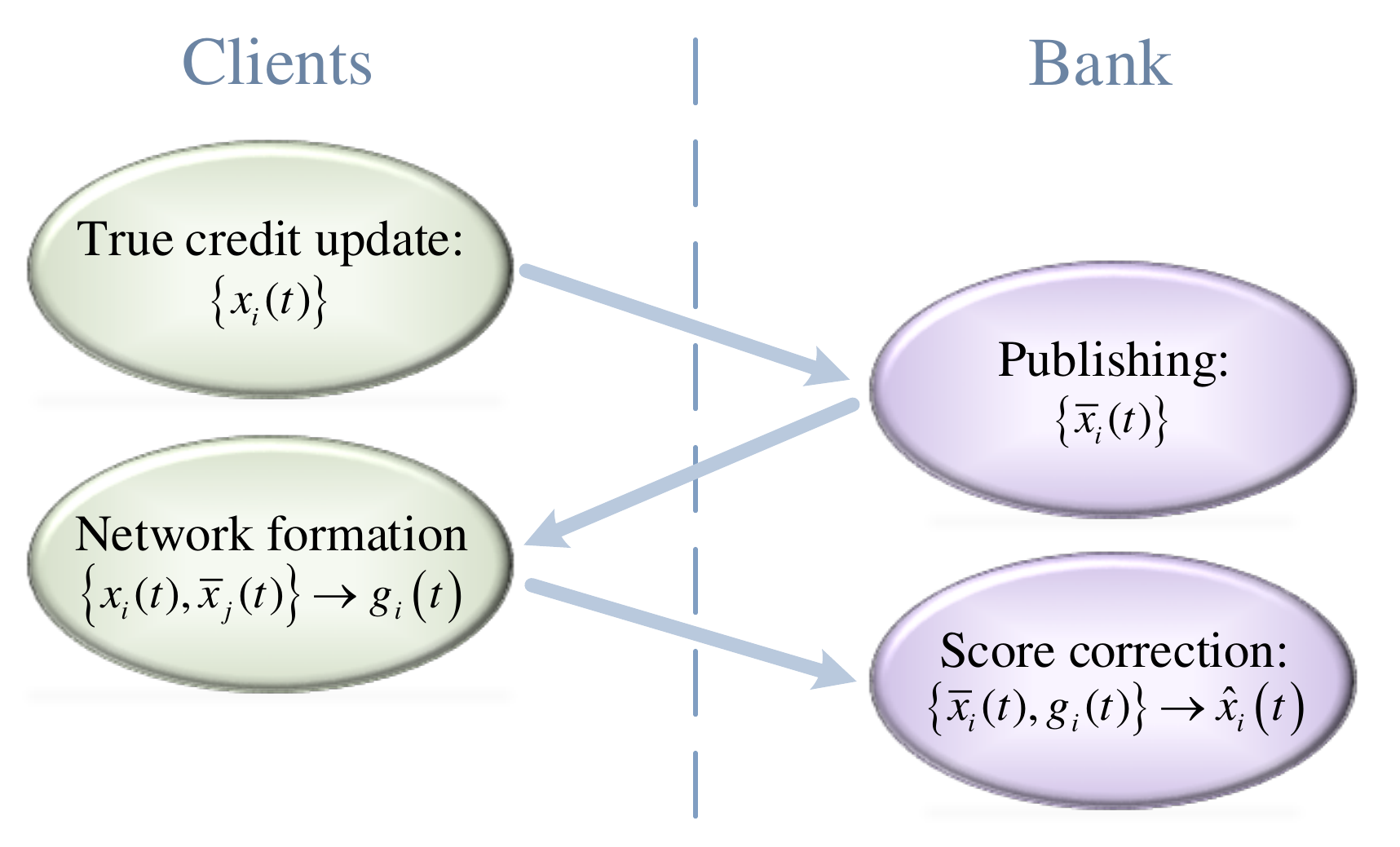}
	\caption{Interaction process at period $t$. }
	\label{Fig_interaction}
\end{figure}

Note that at the beginning, the lender gives an initial estimate $\hat x_i(0)$ based only on individual financial history of clients.
In accordance with Section \ref{Recursive}, we also have the following assumption about the initial estimate.
\begin{assumption} \label{assump_initial}
	For each client $i$, the initial credit estimate $\hat x_i(0)$ is Gaussian with $\hat x_i(0) \sim \mathcal{N}(x_i(0), \hat P_{i}(0))$.
\end{assumption}
\subsection{One-step optimal estimator}
In this section, the method for the score correction step is investigated. During each period, a one-step optimal estimator $\hat x_i (t)$ is updated based on the observation of the current network $g_i(t)$.

It is well-known that the efficient estimators are those uniformly optimal estimators in the class of unbiased estimators, whose variance realizes the Cram\'er--Rao lower bound (CRLB) \citep{lehmann2006theory}. Later in this section it will be shown that the CRLB for the considered problem is not satisfactory enough and unbiased estimators cannot meet our requirements of estimation precision (the reciprocal of the variance). Therefore, in order to further improve the estimation precision, biased estimators can be considered to realize lower variances.  The Bayes estimator is derived through average risk minimization, which could also realize a lower MSE than efficient estimators for clients with creditworthiness in the middle class.


In order to derive the optimal estimate, we define an average risk measure
\begin{equation}\label{e_averRisk}
r(\hat x_i, \alpha) = \int_{ - \infty }^{ + \infty } {\mathbb{E} [L\left( {{x_i},{{\hat x}_i}} \right)]{\alpha}\left( {{x_i}} \right)d{x_i}},
\end{equation}
where $L\left( {{x_i},{{\hat x}_i}} \right)$ is a risk function and ${{\alpha}\left( {{x_i}} \right)}$ is a positive weighting function indicating how important it is to have a low risk for different values of $x_i$.

Without loss of generality, ${{\alpha}\left( {{x_i}} \right)}$ can be normalized by
\begin{equation*}
\int_{ - \infty }^{ + \infty } {{\alpha}\left( {{x_i}} \right)d{x_i}}  = 1.
\end{equation*}

Then the optimal estimate with weighting function $\alpha (x_i)$ is obtained by minimizing the average risk in (\ref{e_averRisk}), i.e.,
\begin{equation}\label{e_BayesEst}
\hat x_{i,\alpha} (g_i) =  \mathop {\arg \min }\limits_{\hat x_i} r(\hat x_i, \alpha),
\end{equation}
which is also known as the Bayes estimator.

When we consider the quadratic cost $L\left( {{x_i},{{\hat x}_i}} \right) = \left\| {{x_i}-{{\hat x}_i}} \right\|^2$, (\ref{e_BayesEst}) can be actually regarded as the MSE estimator of $x_i$ given observation $g_i$. It is also quite inspiring to consider the Bayesian interpretation of (\ref{e_BayesEst}). By considering $x_i$ as the outcome of a random variable $X_i$ whose prior distribution is given by $\alpha(x_i)$, the solution to (\ref{e_BayesEst}) turns out to be the posterior mean, i.e.,
\begin{equation} \label{e_post}
\hat x_{i,\alpha} (g_i) = \mathbb{E} [x_i|g_i].
\end{equation}

Under the framework of recursive scoring based on dynamic interaction, at each period, the ``prior knowledge'' $\alpha(x_i)$ is chosen as the probability density function (pdf) of the predicted score $\bar x_i$. Then the Bayes estimator for each client could be derived.

\begin{theorem}
	Under \textit{Assumption} \ref{assump_initial}, the posterior $x_i\left(t \right)|g_i\left(t \right)$ in (\ref{e_post}) is Gaussian at each time $t \geq 0$, i.e., $x_i\left(t \right)|g_i\left(t \right) \sim \mathcal{N}(\hat x_i(t), \hat P_i \left(t \right))$, with
	\begin{equation}\label{e_Bayes}
	\begin{aligned}
	& \hat x_i(t) = \bar x_i(t) + \frac{{\bar P_{i}(t)}}{1+{\bar P_{i}(t)}{n_{i}(t)}} \sum\limits_{j \in {\mathcal{N}_{i}(t)}} (\bar x_j(t) - \bar x_i (t)),\\
	& \hat P_i \left(t \right) = \frac{{\bar P_{i}(t)}}{1+{\bar P_{i}(t)}{n_{i}(t)}},
	\end{aligned}
	\end{equation}
	where
	\begin{align}
	& \bar x_i(t) = a(t-1) \hat x_i(t-1) + bu_i(t-1),\label{eq:prediction1} \\
	& \bar P_{i}(t) = a(t-1)^2 \hat P_i \left(t-1 \right) + Q_{t-1}. \label{eq:prediction2}
	\end{align}
\end{theorem}

\begin{proof}
	Here we prove the $x_i\left(t \right)|g_i\left(t \right)$ is Gaussian  by induction.
	
	To begin with, we have $\hat x_i(0) \sim \mathcal{N}(x_i(0), \hat P_{i,0})$.
	
	We assume that at time $t-1$, it holds that:
	\begin{equation*}
	x_i\left(t-1 \right)|g_i\left(t-1 \right) \sim \mathcal{N}(\hat x_i(t-1), \hat P_i \left(t-1 \right)),
	\end{equation*}
	for some $\hat x_i(t-1)$ and $\hat P_i \left(t-1 \right)$.
	
	Then we will to show that the above also holds at time $t$. By (\ref{e_step1}) and (\ref{e_publish}), the prior distribution of $x_i(t)$ is Gaussian with $ \mathcal{N} (\bar x_i(t), \bar P_{i}(t)) $, where $ \bar x_i(t) $ and $ \bar P_{i}(t)) $ are given by (\ref{eq:prediction1}) and (\ref{eq:prediction2}) respectively. In the sequel, explicit  dependence on `$t$' is omitted for the sake of brevity.
	
	By Bayesian rule, the pdf of $x_i|g_i$ is given by
	\begin{equation*}\label{post_pdf}
	p(x_i|g_i) = \frac{{\alpha}(x_i) p(g_i|x_i)}{p(g_i)}.
	\end{equation*}
	
	Since we only use neighbors' information in $g_i$, the Bernoulli distribution of $g_i$ is given by
	\begin{equation*}
	p(g_i|x_i) = \prod\limits_{j \in {\mathcal{N}_{i}(t)}} {{\nu e^{ - {{{{\left( {{x_i} - {\bar x_j}} \right)}^2}} \over 2}}}}.
	\end{equation*}
	
	Hence, $p(x_i|g_i)$ can be computed by
	\begin{equation*}
	p(x_i|g_i) \propto
	e^{-\frac{(x_i-{\bar x_i})^2}{2 \bar P_{i}(t)}} \cdot
	\prod\limits_{j \in {N_{i}(t)}} {{e^{ - {{{{\left( {{x_i} - {\bar x_j}} \right)}^2}} \over 2}}}}.
	\end{equation*}
	
	Therefore, $x_i(t)|g_i(t)$ must be Gaussian with
	\begin{equation*}
	x_i\left(t \right)|g_i\left(t \right) \sim \mathcal{N}(\hat x_i(t), \hat P_i \left(t \right)),
	\end{equation*}
	where $\hat x_i(t)$ and $\hat P_i (t)$ are given in (\ref{e_Bayes}), which can be obtained by comparing the coefficients.
\end{proof}

In conclusion, at each time $t$, given the observation of the financial network, the credit scores estimation of each client $i$ is updated by an optimal estimator
\begin{equation*}
\hat x_{i,\alpha} (g_i) = \mathbb{E} [x_i(t)| g_i(t)] = \hat x_i(t),
\end{equation*}
which is given by (\ref{e_Bayes}) and is used in the credit correction step in Fig. \ref{Fig_interaction}.

\subsection{Performance analysis}

\subsubsection{Estimation error and precision} 

The performance of the proposed recursive Bayes estimator is analyzed in this part to show the feasibility of incorporating networked data when updating clients' credit scores. It is shown that the uncertainty for individual creditworthiness could be reduced when compared to that based only on individual attributes, and the estimation precision is improved recursively in the dynamic interaction between the lender and the clients.

The MSE is a widely-used criterion to analyze the performance of estimators, which can be split into
\begin{equation*}\label{e_MSE}
\begin{aligned}
MSE[\hat x_i] & = \mathbb{E} [\Vert \hat x_i - x_i \Vert ^2] \\
& = \mathbb{E} [\Vert \hat x_i - \mathbb{E}[\hat x_i] \Vert ^2] + \Vert \mathbb{E}[\hat x_i] - x_i \Vert ^2 \\
& = Var(\hat x_i) + Bias(\hat x_i)^2,
\end{aligned}
\end{equation*}
where $Var(\hat x_i)$ and $Bias(\hat x_i)$ denote the variance and bias of the estimator, respectively.


Among the unbiased estimators, the efficient estimator is uniformly optimal with a minimal variance (and MSE) equal to CRLB. In this paper, a biased estimator is considered to further improve the estimation precision (the reciprocal of the variance). We will then show that the proposed Bayes estimator could realize a strictly smaller variance than CRLB with a bounded bias. It will be shown in Section 4.3.2 that its MSE is also lower than the efficient estimator for clients with credit scores in a certain range.



\begin{lemma}\label{lemma:crlb}
	At each period $t$ when $x_i(t)$ is required to be estimated from $g_i(t)$, the Bayes estimator given in (\ref{e_Bayes}) has a higher estimation precision than all the unbiased estimators, which means $\hat P_i(t)$ is strictly smaller than $CRLB(x_i(t))$.
\end{lemma}

\begin{proof}
	First, the Fisher information matrix (FIM) of $x_i$ computed from the log-likelihood function is
	\begin{equation*}
	I_F(x_i(t)) = -\mathbb{E} \left[\frac{\partial ^2 log(p(g_i(t);x_i(t)))}{\partial x_i(t)^2} \right ] = n_{i}(t).
	\end{equation*}
	
	Then, by (\ref{e_Bayes}), it holds that
	\begin{equation*}
	\hat P_i(t)^{-1} = n_{i}(t) + \bar P_{i}(t)^{-1} > I_F(x_i(t)),
	\end{equation*}
	and thus
	\begin{equation*}
	\hat P_i(t) < I_F(x_i(t))^{-1} = CRLB(x_i(t)).
	\end{equation*}
\end{proof}

For the proposed Bayes estimator, at each step a lower variance is achieved at the cost of a bias. In order to guarantee the accuracy of score estimation in the dynamic process, we will then show that its bias is bounded throughout the whole time interval. 

\begin{lemma}\label{lemma:state}
	(Bounds of prediction value) Assuming $|\bar{x}_i(0)| \leq M_0$, where $M_0 = \max \limits_{i=1,\ldots N} |\bar{x}_i(0)|$, the prediction value of the credit scorings given by equation \eqref{eq:prediction1} is bounded if $0<a(t) < 1$, for any $t \geq 0$.
\end{lemma}

\begin{proof}
	According to  \eqref{eq:prediction1}, we have
	\begin{equation*}
	\hat{x}_i(t) = \frac{1}{a(t)} (\bar x_i(t+1) - b(t)u_i(t)).
	\end{equation*}
	Thus,
	\begin{equation*}
	\frac{1}{a(t)} (\bar x_i(t+1) - b(t)u_i(t)) = \bar x_i(t) + \frac{{\bar P_{i}(t)}}{1+{\bar P_{i}(t)}{n_i(t)}} \sum\limits_{j \in {\mathcal{N}_{i,t}}} (\bar x_j(t) - \bar x_i (t)),\\
	\end{equation*}
	then
	\begin{small}
		\begin{equation*}
		\begin{split}
		\bar x_i(t+1)  &= a(t) [\bar x_i(t) + \frac{{\bar P_{i}(t)}}{1+{\bar P_{i}(t)}{n_i(t)}} \sum\limits_{j \in {\mathcal{N}_{i,t}}} (\bar x_j(t) - \bar x_i (t))] + b(t) u_i(t)\\
		& = a(t) [ \frac{1}{1+{\bar P_{i}(t)}{n_i(t)}} \bar x_i(t) + \frac{{\bar P_{i}(t)}}{1+{\bar P_{i}(t)}{n_i(t)}} \sum\limits_{j \in {\mathcal{N}_{i,t}}} \bar x_j(t)] + b(t) u_i(t).
		\end{split}
		\end{equation*}
	\end{small}
	Considering the first step, we have that
	\begin{small}
		\begin{equation*}
		\begin{split}
		|\bar x_i(1)|  & = |a(0) [ \frac{1}{1+{\bar P_{i,0}}{n_{i,0}}} \bar x_i(0) + \frac{{\bar P_{i,0}}}{1+{\bar P_{i,0}}{n_i(t)}} \sum\limits_{j \in {\mathcal{N}_{i,t}}} \bar x_j(0)] + b(0) u_i(0)|\\
		& = |a(0) [ \frac{1}{1+{\bar P_{i,0}}{n_{i,0}}} \bar x_i(0) + \frac{{\bar P_{i,0}}}{1+{\bar P_{i,0}}{n_i(t)}} \sum\limits_{j \in {\mathcal{N}_{i,t}}} \bar x_j(0)] + b(0) u_i(0)|\\
		&\leq |a(0) M_0 + b(0)u_i(0)| \\
		&\leq |a(0) M_0 |+ |b(0)u_i(0)|.
		\end{split}
		\end{equation*}
	\end{small}
Iterating, it holds that
	\begin{equation*}
	|\bar x_i(t)|  \leq M_0 \prod\limits_{k = 0}^{t - 1} {a(k)}  + \sum\limits_{k = 0}^{t - 1} {\left[ {\left( {\prod\limits_{l = k + 1}^{t - 1} {a(l)} } \right)\left| {b\left( k \right)u\left( k \right)} \right|} \right]} .
	\end{equation*}
	
	Since $0<a(t)< 1$, $\bar x_i(t)$ can be bounded by a constant $M$ for any $t \geq 0 $, i.e., $|\bar x_i(t)| \leq M$.
\end{proof}

\begin{assumption}\label{assu:Q}
	$Q_t$ is bounded with a lower bound and an upper bound $Q_{l}$, $Q_{u}$, respectively, i.e., $ Q_{l} \leq Q_t \leq Q_{u}$ for any $t>0$.
\end{assumption}

\begin{theorem}[Bound of estimation precision $\hat{P}_{i}(t)$]
	The estimation precision $\hat{P}_{i}(t)$ is bounded by a lower and upper bound, respectively, i.e., $ {P}_{l}(t) \leq \hat{P}_{i}(t) \leq {P}_{u}(t) < CRLB(x_i(t))$.
\end{theorem}

\begin{proof}
	First, we prove $\hat{P}_{i}(t)$ has a lower bound. According to  \eqref{e_Bayes},
	\begin{equation*}
	\bar P_{i}(t) = a(t-1)^2 \hat P_i \left(t-1 \right) + Q_{t-1} \geq Q_{l},
	\end{equation*}
	thus $\frac{1}{\bar P_{i}(t)} \leq  \frac{1}{Q_{l}} $.
	From \eqref{e_Bayes}, we have
	\begin{equation*}\label{eq:hatP}
	\frac{1}{\hat P_i \left(t \right)} = \frac{1}{\bar P_i \left(t \right)} +  n_i(t) \leq \frac{1}{Q_{l}} + N,
	\end{equation*}
	where $N$ is the number of clients, denoting $ P_l = (\frac{1}{Q_{l}} + N)^{-1}$, we have
	$\hat P_{i}(t) \geq P_l$.

	Secondly, we prove that $\hat{P}_{i}(t)$ has an upper bound. According to Eq. \eqref{e_Bayes}, \eqref{eq:prediction2}, and Assumption \ref{assu:Q}, we have
	\begin{equation*}\label{eq:P1}
	\begin{split}
	\hat P_i \left(t \right) & =   \frac{a(t-1)^2 \hat P_i \left(t-1 \right) + Q_{t-1}}{1 + (a(t-1)^2 \hat P_i \left(t-1 \right) + Q_{t-1})n_i(t) }. \\
	\end{split}
	\end{equation*}
	Hence
	\begin{equation*}\label{eq:P3}
	\begin{split}
	\frac{1}{\hat P_i \left(t \right)} &=  \frac{1}{\hat P_i \left(t-1 \right)} \frac{1}{(a(t-1)^2 + \hat P_i^{-1} \left(t-1 \right)Q_{t-1})}   +  n_i(t) \\
	& \geq \frac{1}{\hat P_i \left(t-1 \right)} \frac{1}{a(t-1)^2 + (Q_l^{-1} + N)Q_u} + n_i(t) \\
	& \geq m_0^t  \frac{1}{\hat P_i \left(0 \right)} + \sum_{k=0}^{t}  m_0^k n_{i,t-k},
	\end{split}
	\end{equation*}
	where $m_0 = \frac{1}{{\bar a}^2 + (Q_l^{-1} + N)Q_u}$ and $\bar a = \mathop {\max }\limits_{0 \leq k \leq t} a(k) $.
	
	Therefore,
	\begin{equation*}\label{eq:P4}
	\begin{split}
	\hat P_i \left(t \right) &\leq (m_0^t  \frac{1}{\hat P_i \left(0 \right)} + \sum_{k=0}^{t}  m_0^k n_{i}(t-k))^{-1},
	\end{split}
	\end{equation*}
	we  see that $\hat P_i \left(t \right)$ is bounded by $P_u = (m_0^t  \frac{1}{\hat P_i \left(0 \right)} + \sum_{k=0}^{t}  m_0^k n_{i}(t-k))^{-1}$, which is strictly smaller than $CRLB(x_i(t))$.
\end{proof}

\subsubsection{A special case}\label{section:special} 

When only unbiased estimators are considered, efficient estimators are optimal with MMSE equal to CRLB.
However, when all estimators are taken into account, there does not exist an estimator that is uniformly optimal for all values of $x_i \in [0, M]$. One reason is that the bias is always dependent on the specific value of the parameter to be estimated. Therefore, the Bayes estimator is designed to realize a better estimation performance for a subset of clients. In this part, such intuition is illustrated by a special case, where the proposed estimator is able to realize a lower MSE than all efficient estimators for clients in the middle class.

\begin{assumption}\label{ass:uni}
Consider a short period where the true credit scores are assumed to be constant, i.e., $x_i = x_i(0)$, for any $t \geq 0$, and $\lbrace x_i \rbrace_{i=1}^N$ is uniformly distributed on $[0, M]$ with $M > 3$.
\end{assumption}

\begin{theorem}\label{uniform}
	 For the middle-class clients with $x_i(t) \in [3 , M-3]$, the estimator in (\ref{e_Bayes}) is unbiased when $N \to +\infty$ (which also indicates consistency of the estimator) and the corresponding MSE is lower than that of the efficient estimator, i.e $MSE(\hat x_i(t)) \leq CRLB(x_i(t))$ under Assumption \ref{ass:uni}.
\end{theorem}

\begin{proof}
	To begin with, we consider the first step
	\begin{equation*}
	\begin{aligned}
	& \bar x_i(1) = a(0) \hat x_i(0) + b(0)u_i(0), \\
	& \hat x_i(1) = \bar x_i(1) + \frac{{\bar P_{i}(1)}}{1+{\bar P_{i}(1)}{n_{i}(1)}} \sum\limits_{j \in {\mathcal{N}_{i,1}}} (\bar x_j(1) - \bar x_i (1)),
	\end{aligned}
	\end{equation*}
	for $x_i(1) \in [3 , M-3]$.
	
	Under \textit{Assumption} \ref{assump_initial}, we  see that $\bar x_j(1)$ is unbiased with $\bar x_j(1) \sim \mathcal{N} (x_j(1), {\bar P_{j}(1)})$ for any $j=1,\ldots,N$, which can also be rewritten as:
	\begin{equation*}
	\bar x_j(1) = x_j(1) + e_j(1), \quad \text{with}~ x_j(1) \sim \mathcal{U}[0, M] ~\text{and}~ e_j(1) \sim \mathcal{N} (0, {\bar P_{j}(1)}).
	\end{equation*}
	
	Thus, the pdf of $\bar x_j(1)$ can be given by the convolution of two density functions as
	\begin{equation*}
	\begin{aligned}
	p_{\bar X}(\bar x_j(1)) &= \int_{-\infty}^{+\infty} p_X(x_j(1))p_E(\bar x_j(1) - x_j(1))d  x_j(1) \\
	&=\dfrac{1}{2M} \left[erf(\dfrac{M-\bar x_j(1)}{\sqrt{2{\bar P_{j}(1)}}}) - erf(\dfrac{-\bar x_j(1)}{\sqrt{2{\bar P_{j}(1)}}})\right],
	\end{aligned}
	\end{equation*}
	where $erf$ denotes the error function.
	
	Recall that $\mathcal{N}_i(1)$ is generated based on (\ref{e_PrBer}), 
	we  see that $\mathbb{E}[n_i(1)] \to  + \infty$ as $N \to  + \infty$. Then by the law of large number, for any given $x_i(1) \in [3 , M-3]$, it holds that
	\begin{equation*}
	\begin{aligned}
	\mathop {\lim }\limits_{N \to  + \infty } \mathbb{E}[\hat x_i(1)] - x_i(1) &= \mathbb{E}[\bar x_j(1) - \bar x_i(1)] \\
	&= \int_{-\infty}^{+\infty} \nu e^{-\frac{(\bar x_j(1) - x_i(1))^2}{2}}(\bar x_j(1) - x_i(1))p_{\bar X}(\bar x_j(1)) d \bar x_j(1) \\
	&= \int_{x_i(1)-3}^{x_i(1)+3} \nu e^{-\frac{(\bar x_j(1) - x_i(1))^2}{2}}(\bar x_j(1) - x_i(1))p_{\bar X}(\bar x_j(1)) d \bar x_j(1) \\
	&= 0,
	\end{aligned}
	\end{equation*}
	where the last line results from the three-sigma rule and the fact that $p_{\bar X}(\bar x_j(1))$ is identical on $[3 , M-3]$ since $\bar P_j(1)$ is far smaller than $M$.
	
	In this situation, we have that $\hat P_i(t+1) = \frac{\hat P_i(t)}{1+\hat P_i(t) n_i(t)}$ is decreasing with $t$. Then the recursion can be conducted in the same manner as shown above, thus leading to an unbiased estimator for any $t \geq 0$, i.e.,
	\begin{equation*}
	\mathop {\lim }\limits_{N \to  + \infty } \mathbb{E}[\hat x_i(t)] = x_i(t), \quad \text{for any} \quad t \geq 0,
	\end{equation*}
	for any $x_i \in [3 , M-3]$.
	
	Hence,
        \begin{equation*}
	MSE(\hat x_i(t)) = \hat p_i(t) < CRLB(\hat x_i(t)), \quad \text{for any} \quad t \geq 0,
	\end{equation*}
	which is also converging to zero with probability almost 1.
	
\end{proof}

\begin{remark}
	As for the clients with credit scores near the boundary, it is obvious that the proposed estimator is biased. For example, the score estimate for an underprivileged client will be lifted since most of its neighbors have a higher score. Recall that the observation that merely relying on individual assets is unbiased with a high variance. Then the Bayes estimator realizes a lower variance at the expense of bias error, which can be considered as a trade-off between individual attributes and networked information. Such trade-off between bias and variance can be considered by the lender by adjusting the weighting function $\alpha (x_i)$ in (\ref{e_averRisk}). Alternatively, for those clients with extremely low or high credit scores, the lender can also choose to make assessments based only on individual financial attributes and consider the networked information as a reference for risk prediction as shown in Section \ref{Recursive}.
\end{remark}

\section{Numerical simulations}\label{section:Numerical Simulations}
In this section, some numerical examples are given to illustrate the performance of the proposed credit scoring algorithms using dynamic networked information. Firstly, estimation errors and variances for the Bayes estimator of 50 clients are presented. Secondly, Monte Carlo simulations are conducted to illustrate the effectiveness and consistency of the proposed algorithm. Our simulation results validate the theoretical analysis discussed above.

Here we consider the special case as given in Section \ref{section:special}, where the parameters in \eqref{e_step1} are chosen to be $a = 1$, $ b = 0$, $ Q_t = 0$, respectively, i.e., we highlight the effect of the networked information provided that everything else is equal. The scoring process is conducted on time interval $t = 1,\cdots, 15$. The scenario of online scoring is considered. Initially, the lender can only obtain a noisy credit estimate
based on the limited information of individual assets.
At each time, every client forms a new homophily-based network according to others' credit reports published by the lender, which is then used by the lender to make a one-step optimal predictor for the next period.
Under such framework of dynamic interaction, the links among clients are reconstructed at each period with the updated  publishing of the scores by the lender.
Fig. \ref{fig:trading} presents the network structure at time step $t = 15$. It is a directed network, where clients with similar credit scores are connected with a probability of $Pr(g_{ij}(t) = 1) = e^{-\frac{(x_i(t) - \bar x_j(t))^2}{2}}.$

\begin{figure*}
	\centering
		\includegraphics[width=8cm]{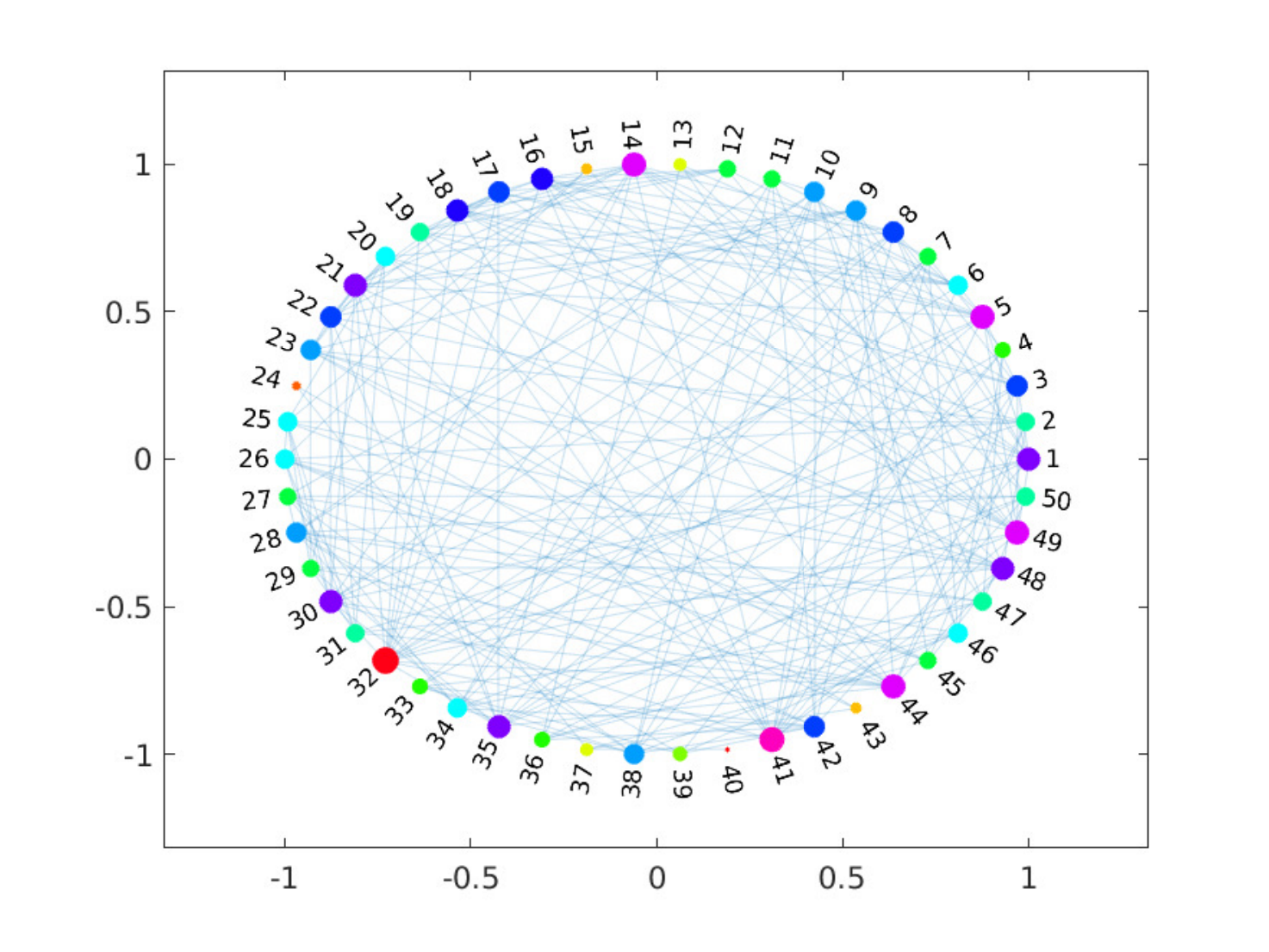}    
		\centering{\caption{Network among $50$ clients at time step $t = 15$}
			\label{fig:trading}}
\end{figure*}

The recursive Bayes estimator based on dynamic interaction in Section \ref{dynamic} is used for lenders to update the credit scores for each client. Estimation results and the variance of the errors are shown in Figure \ref{fig:Estimation}. The real state (red dots in Fig. \ref{fig:trading}) stands for the true credit states for each client, which are presented in an ascendant order. The estimation values at time step $t=1$, $t=5$, $t=15$ are presented for each client, respectively. We  see that for the middle-class clients, i.e., the clients with credit scores around $4 \sim 12$ in our simulation, the estimation results converge to their true values in a few steps, while the underprivileged clients and privileged clients (clients with credit scores smaller than $4$ and  bigger than $12$, respectively) have a positive and negative bias, respectively. Such conclusion is consistent with the theoretical results in Section \ref{section:special}. The variance of the estimation errors on the right of Fig. \ref{fig:Estimation} shows that $\hat P_i (t)$ is decreasing to the lower bound $P_l$. Furthermore, since there is no system noise $w(t)$, $\hat P_i (t)$ will converge to $0$ with probability of almost 1.

\begin{figure}
	\centering
	\subfigure{
		\begin{minipage}[t]{0.5\linewidth}
			\centering
			\includegraphics[width=8cm]{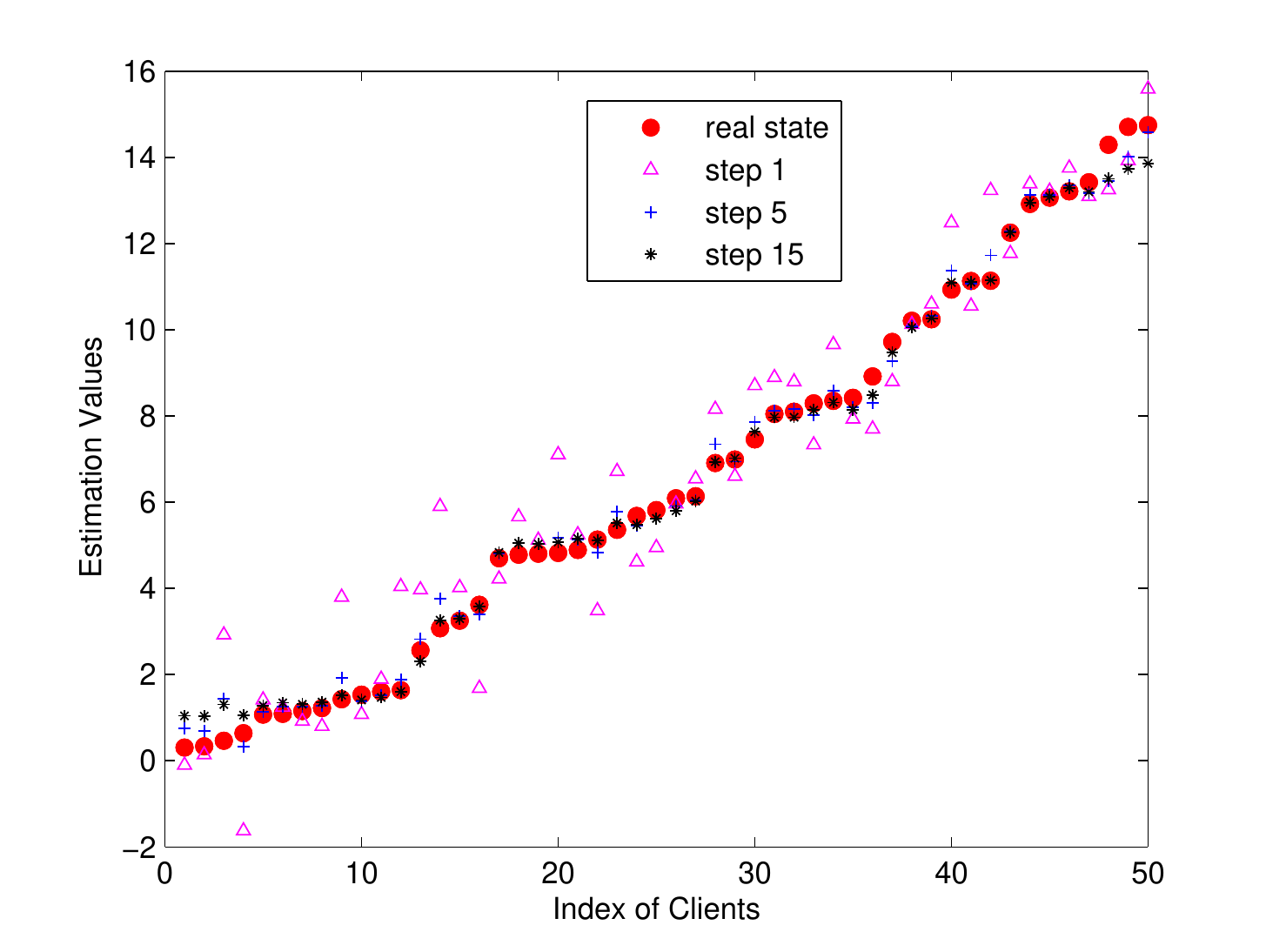}
		\end{minipage}%
	}%
	\subfigure{
		\begin{minipage}[t]{0.5\linewidth}
			\centering
			\includegraphics[width=8cm]{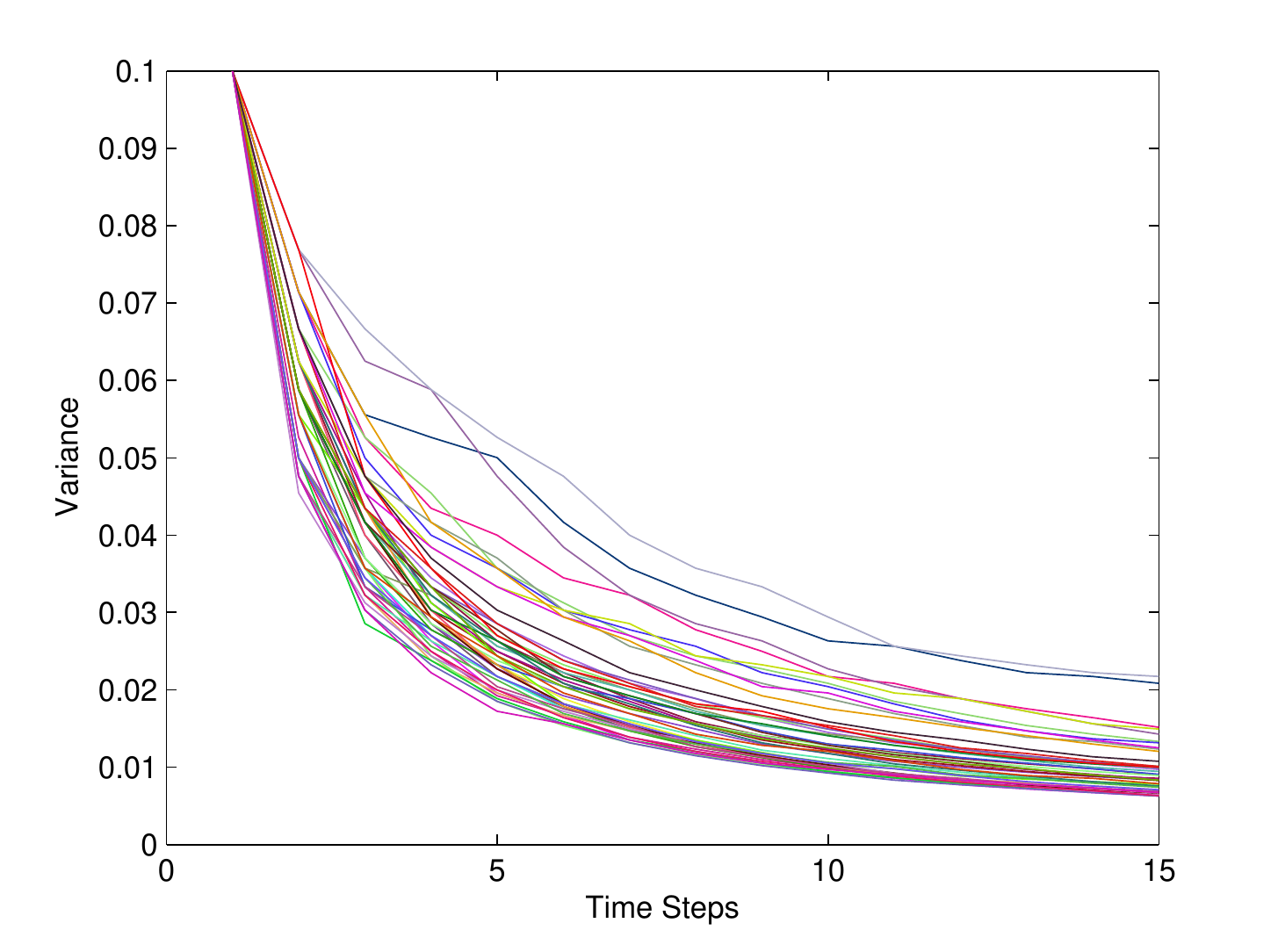}
		\end{minipage}%
	}%
	\caption{Credit scoring estimation (left) and the variance of error  (right)}
	\label{fig:Estimation}
\end{figure}

Besides, MSE and CRLB of the recursive Bayes estimator are compared in Figure \ref{fig:MSE} to show the effectiveness of our algorithm. The results of step $t=10$ and step $t=15$ are presented, respectively. We see that MSE is lower than CRLB for middle-class clients for these two cases. Meanwhile, since the estimation accuracy is improved recursively and the proposed estimator is consistently unbiased for scores in the middle range, their mean square errors are converging with increasing iterations.

\begin{figure}
	\centering
	\subfigure{
		\begin{minipage}[t]{0.5\linewidth}
			\centering
			\includegraphics[width=8cm]{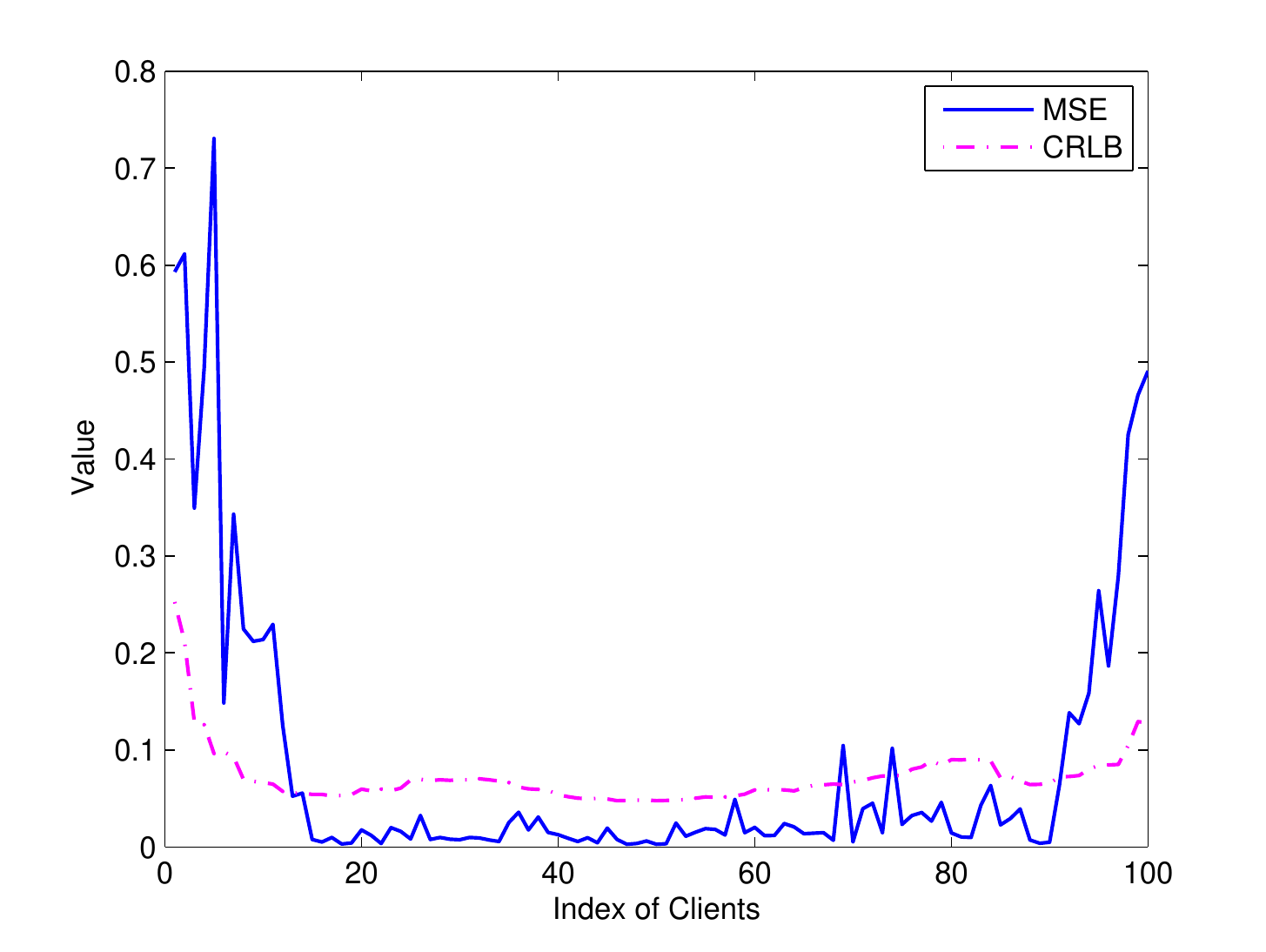}
		\end{minipage}%
	}%
	\subfigure{
		\begin{minipage}[t]{0.5\linewidth}
			\centering
			\includegraphics[width=8cm]{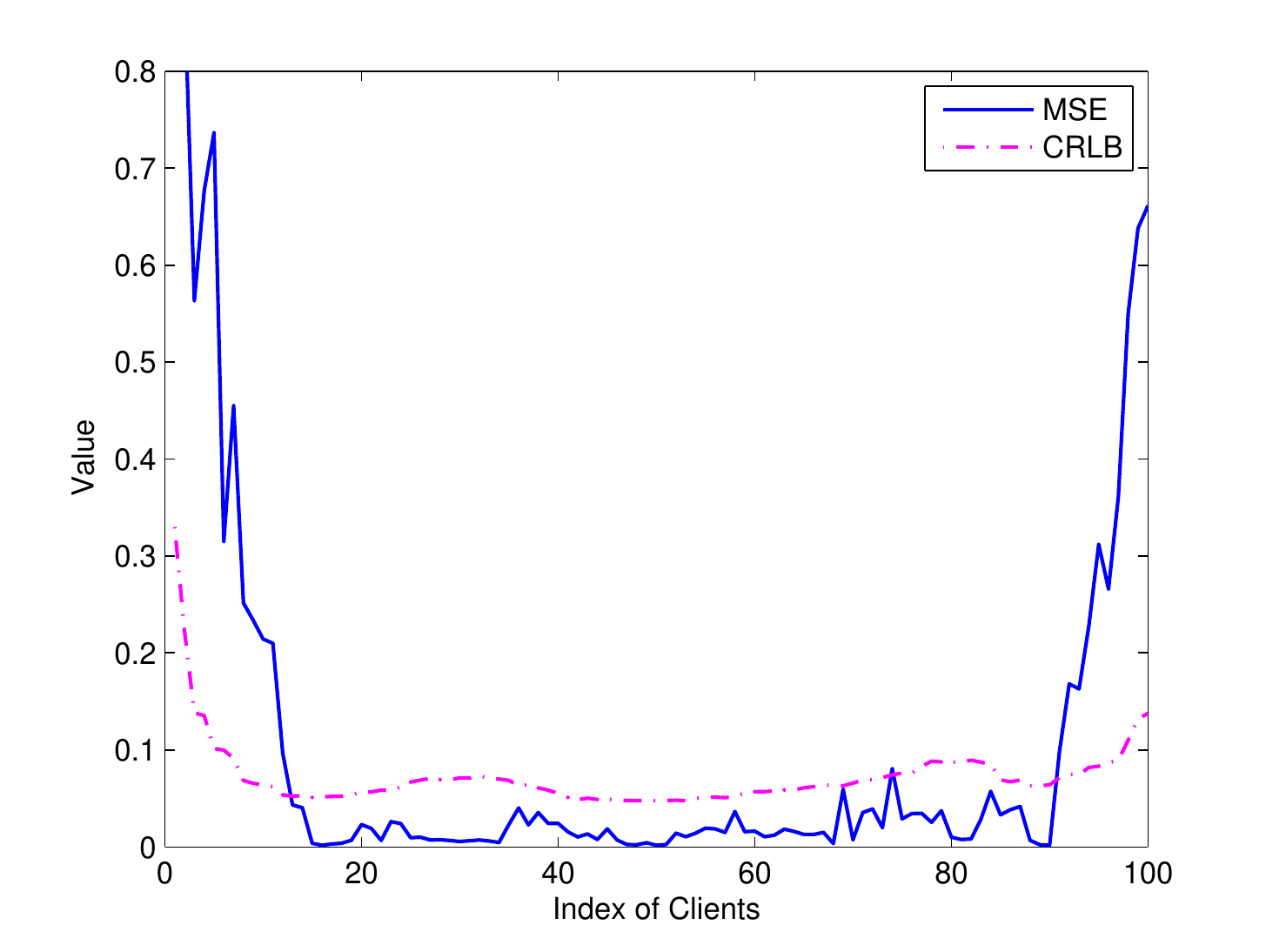}
		\end{minipage}%
	}%
	\caption{MSE and CRLB for $N=100$, eatimation step $t = 10$ (left) and and $t = 15$  (right)}
	\label{fig:MSE}
\end{figure}

Moreover, in order to illustrate the consistency of the proposed algorithm, the error box (\cite{mcgill1978variations}) for Monte Carlo simulations are plotted to show statistical properties (the number of Monte Carlo is 100). In Figure \ref{fig:Error}, the relative error ($\hat{x}(t) - x(t)$) of at the time step $t = 15$ estimation for each client are presented in forms of error box. Two networks with the number of clients $N=50$ and $N=100$ respectively are compared. Particularly, on each box for client, the central mark indicates the median of the relative error, and the bottom and top edges of the box indicate the $25th$ and $75th$ percentiles of the relative error, respectively. 
The outliers are plotted individually using the '+' symbol. From the error box for $N=50$ in the left plot we  see that the estimations for client 10 to client 40 are more accurate than others since the relative errors fluctuate around $0$. Furthermore, the magnitudes of their estimation errors are decreasing when they have more connections in a bigger network.  
This result is in accordance with the consistency of the estimator proved in Theorem \ref{uniform}. Moreover, we also observe that clients with lower credit scores are overestimated whereas those with higher credit scores are underestimated to a certain extent, which illustrates that incorporating networked information could impose different effects on different types of clients.
\begin{figure}
	\centering
	\subfigure{
		\begin{minipage}[t]{0.5\linewidth}
			\centering
			\includegraphics[width=8cm]{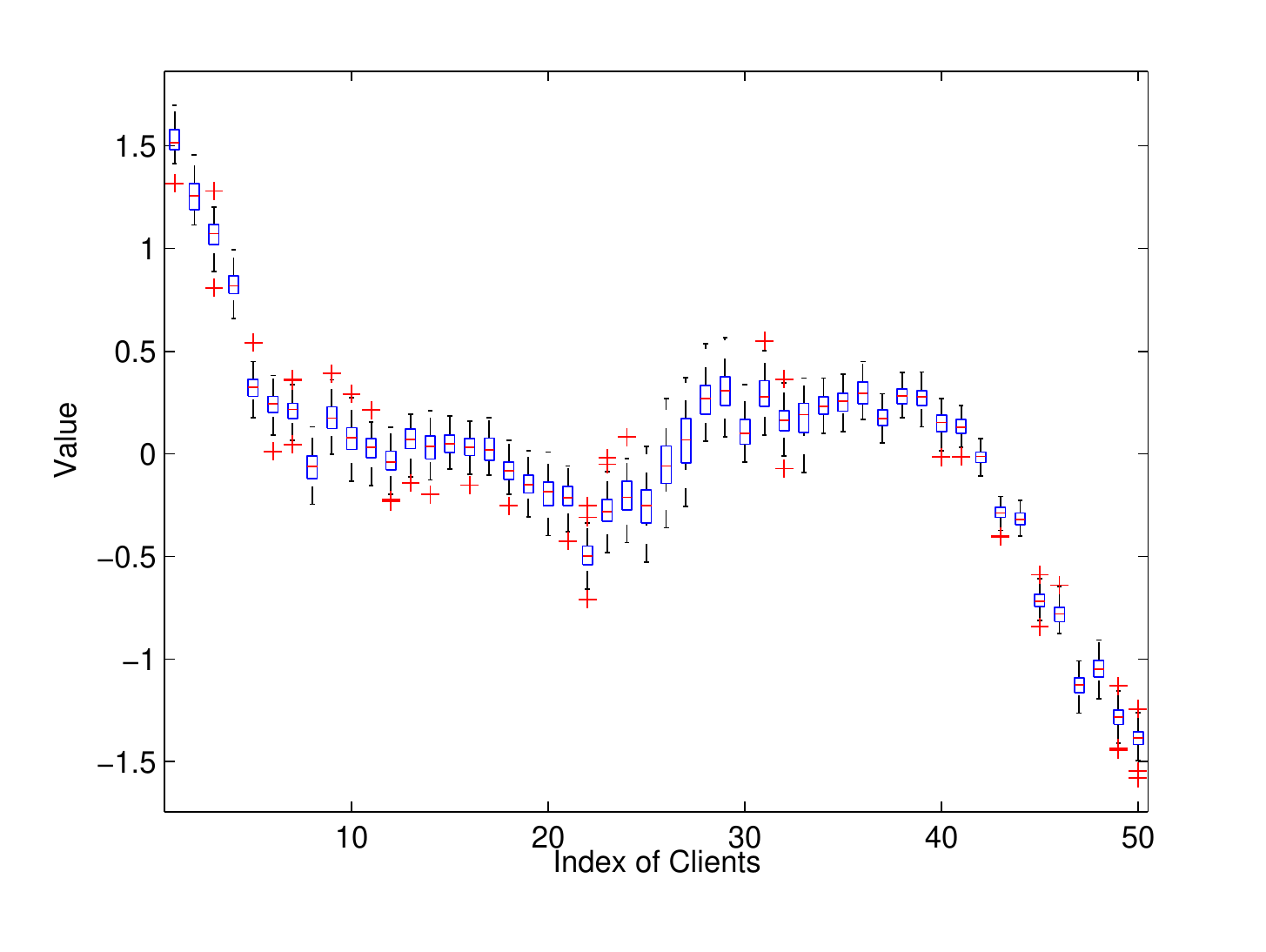}
		\end{minipage}%
	}%
	\subfigure{
		\begin{minipage}[t]{0.5\linewidth}
			\centering
			\includegraphics[width=8cm]{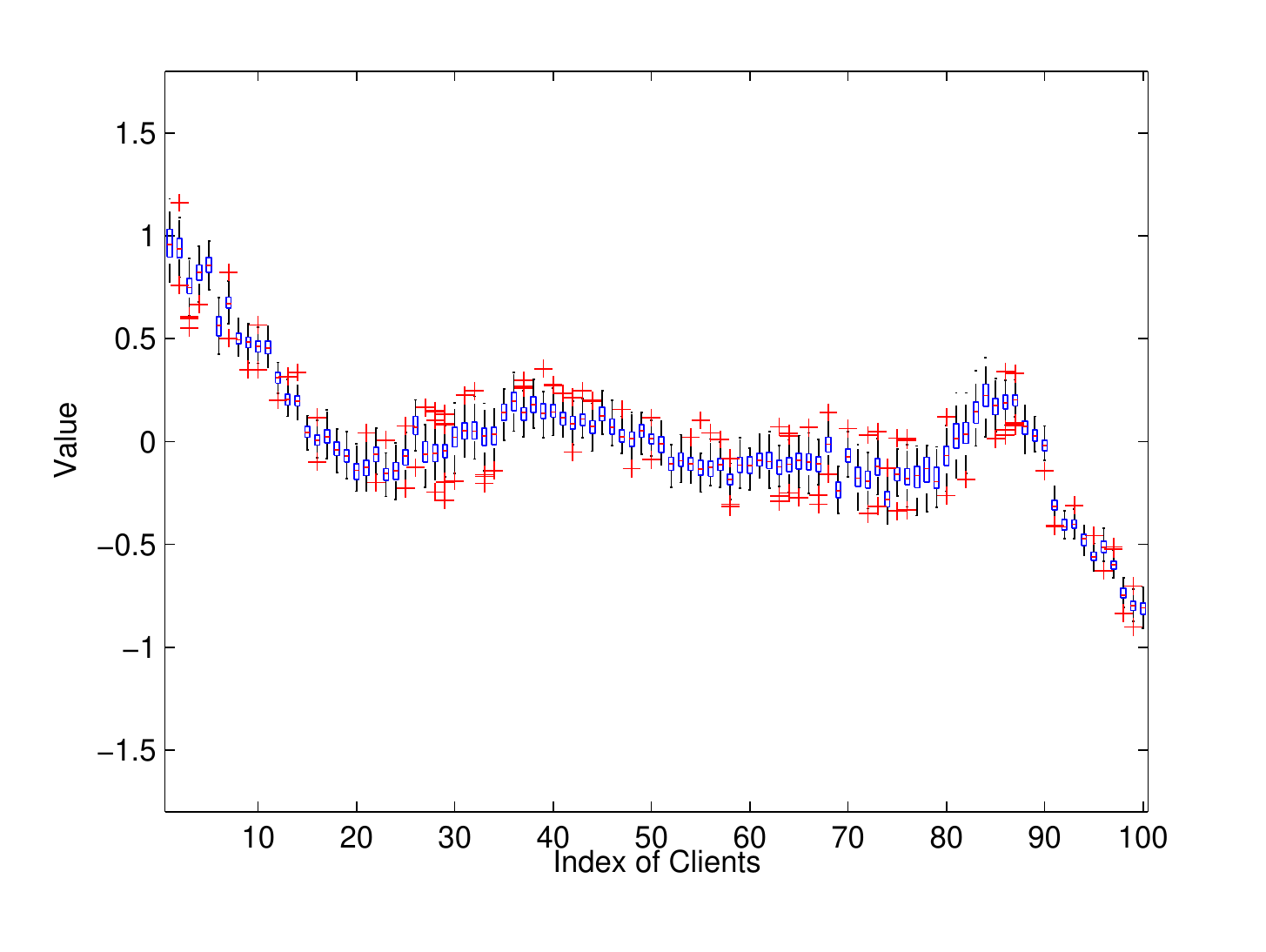}
		\end{minipage}%
	}%
	\caption{Error plot for numbers of clients $N=50$ (left) and $N=100$  (right)}
	\label{fig:Error}
\end{figure}

\section{Discussion and conclusion}\label{section:Conclusions and Future Work}

\subsection{Discussion}
The mathematical analysis in this paper shows that networked data can provide additional relevant information for risk management and clients' true creditworthiness. 
The benefits of incorporating networked data are more significant for societies that have a higher density of interactions. As shown in (\ref{P_risk}) and (\ref{e_Bayes}), the precision of estimation is improved to a larger extent for clients with higher degree centrality (number of connections). As many companies have started collecting networked data, the results obtained in this paper could potentially help them in terms of data collection and decision making, namely, to decide which kind of networked data is actually relevant and informative for a specific society, and how to incorporate it in the existing framework. 

The results of this paper also suggest that different measures could be considered for different types of clients. For example, for new clients with limited financial history, creditworthiness merely based on individual attributes can be rather inaccurate, thus their social or trading networks ought to be incorporated by the lenders to reduce risk in the near future.
Our results further indicate that incorporation of networked data in credit scoring has different effects on different types of clients. With traditional credit scoring methods that only rely on individual attributes, it is usually difficult for underprivileged clients to get loans. Such problems also concern the existing discussion on inequality in the financial sector. It can be observed from (\ref{e_Bayes}) that their connections with other clients would increase their opportunities for better financial support. However, under such practice of credit scoring, there is a risk that clients with high scores might be reluctant to establish connections with underprivileged clients in order to keep their scores, which would result in kind of financial segregation problem due to adverse selection. Therefore, more attention is required for the lenders as well as the regulators to balance those effects, and different scoring methods could be considered for different types of clients.

\subsection{Conclusion}
Currently, it is common practice to use the structured financial data such as loan characteristics (purpose of the loan and its duration), clients characteristics (credit utilization) and credit history (repayment of previous loans) for estimating the credit score.  Motivated by increasing empirical practices with networked data in the financial sector, this paper investigates from a theoretic point of view the relevance of networked data to individual creditworthiness. A model-based framework is proposed to address this problem systematically by incorporating networked information in a dynamic environment. The interaction between clients is modelled based on the assumption of homogeneous preference. Then it is shown by Bayesian approaches how networked data could facilitate financial decision making in two scenarios respectively.
Firstly, we propose a Bayesian optimal filter to predict clients' credit scores if the publishing of the credit scores are estimated merely from structured financial data. Such prediction is used as a monitoring indicator for the risk warning  in lenders'
future financial decisions. Secondly, we develop a recursive Bayes estimator to improve the precision of score estimation by incorporating network topologies as well. It is shown that under the proposed evolution framework, the designed estimator has a higher precision than any efficient estimator, and the mean square errors are strictly smaller than the Cram\'er--Rao lower bound for clients within a certain range of scores. For further investigation, simulation results for a case where true credits are uniformly distributed illustrate the effectiveness and feasibility of the proposed methods.



\bibliographystyle{cas-model2-names}

\bibliography{ref}

\end{document}